\documentclass[12pt]{article}

\usepackage{tikz}
\usepackage{graphicx}
\usepackage{verbatim}
\usepackage{url}
\usepackage{fullpage}
\usepackage{amssymb,amsfonts,amsmath,amsthm}
\usepackage[colorlinks=true,linkcolor=blue,citecolor=red]{hyperref}
\usepackage{algorithm}
\usepackage{algpseudocode}
\newcommand{\polylog}{\mathrm{polylog}}

\newcommand{\floor}[1]{{\lfloor#1\rfloor}}
\newcommand{\ceil}[1]{{\lceil#1\rceil}}

\newtheorem{theorem}{Theorem}[section]
\newtheorem{corollary}[theorem]{Corollary}
\newtheorem{proposition}[theorem]{Proposition}
\newtheorem{definition}[theorem]{Definition}
\newtheorem{lemma}[theorem]{Lemma}

\title{Sorting Networks On Restricted Topologies}

\newcommand{\FormatAuthor}[2]{
\begin{tabular}{c}
#1  \\ {\small #2}
\end{tabular}
}
\author{
\begin{tabular}[h!]{lcr}
   \FormatAuthor{Avah Banerjee}{George Mason University}
&   \FormatAuthor{Dana Richards}{George Mason University}
&   \FormatAuthor{Igor Shinkar}{UC Berkeley}
\end{tabular}
}

\begin{document}

\maketitle

\begin{abstract}
The \textit{sorting number} of a graph with $n$ vertices is the minimum depth
of a sorting network with $n$ inputs and $n$ outputs that uses only the edges
of the graph to perform comparisons.
Many known results on sorting networks can be stated in terms of sorting
numbers of different classes of graphs. In this paper we show the
following general results about the sorting number of graphs.
\begin{enumerate}
  \item Any $n$-vertex graph that contains a simple path of length $d$ has a sorting network of depth $O(n \log(n/d))$.
  \item Any $n$-vertex graph with maximal degree $\Delta$ has a sorting network of depth $O(\Delta n)$.
\end{enumerate}
We also provide several results relating the sorting number of a graph with its routing number,
size of its maximal matching, and other well known graph properties. Additionally, we give
some new bounds on the sorting number for some typical graphs.
\end{abstract}

\section{Introduction}

In this paper we study oblivious sorting algorithms.
These are sorting algorithms whose sequence of comparisons is made in advance, before seeing the input,
such that for any input of $n$ numbers the value of the $i$'th output is smaller or equal to
the value of the $j$'th output for all $i < j$.
That is, for any permutation of the input out of the $n!$ possible, the output of the algorithm must be sorted.
A sorting network, which typically arises in the context of parallel algorithms,
is an oblivious algorithm where the comparisons are grouped into {\em stages},
and in each stage the compared pairs are disjoint.
In this paper we explore the situation where a given graph specifies which keys are allowed to be compared.
We regard a sorting network as a sequence of stages, where each stage corresponds to a matching in the graph
and a comparator is assigned to each matched pair.
There are fixed locations, each containing a key,
and a comparator looks at the keys at the endpoints of the edges of the matching,
and swap them if they are not in the order desired by the underlying oblivious algorithm.
Therefore, we say that the underlying algorithm induces a {\em directed matching}.
The locations are ordered, and the goal is to have the order of the keys match the order of the locations after the execution of the algorithm.
The {\em depth} of a sorting network is the number of stages, and the {\em size} is the total number of edges in all the matchings.
Note that for an input of length $n$
at most $\floor{n/2}$ comparisons can be performed in each step,
and hence the well-known lower bound of $\Omega(n \log(n))$ comparisons
in the sequential setting implies a $\Omega(\log(n))$ lower bound
on the depth on the network, that is, the number of stages in the network.

A large variety of sorting network have been studied in the literature.
In their seminal paper, Ajtai, Koml{\'o}s, and Szemer{\'e}di~\cite{ajtai19830}
presented a construction of a sorting network of depth $O(\log {n})$. We will refer to it
as the \emph{AKS sorting network}.
In this work we explore the question of constructing a sorting network
where we are given a graph specifying which keys are allowed to be compared.
We define the {\it sorting number} of a graph $G$, denoted by $st(G)$,
to be the minimal depth of a sorting network that uses only the edges of $G$.
The AKS sorting network can be interpreted as a sorting network on the complete graph,
i.e., $st(K_n) = O(\log(n))$.
More precisely, the AKS construction specifies \emph{some} graph $G_{AKS}$ whose maximal degree is $O(\log(n))$
and $st(G_{AKS}) = O(\log(n))$.

The setting where the comparisons are restricted to the $n$-vertex path graph, denoted by $P_n$,
is perhaps the easiest case. It is well known that $st(P_n) = n$,
which follows from the fact that the odd-even transposition sort takes $n$ matching steps
(see, e.g.,~\cite{knuth1998art}).
For the hypercube graph $Q_d$ on $n = 2^d$ vertices we can use the Batcher's bitonic sorting network,
which has a depth of $O((\log n)^2)$~\cite{batcher1968sorting}.
This was later improved to $2^{O(\sqrt{\log\log n})}\log n$ by Plaxton and Suel~\cite{leighton1998hypercubic}.
We also have a lower bound of $ \Omega(\frac{\log n \log\log n}{\log \log \log n})$
due to Leighton and Plaxton~\cite{plaxton1994super}.
For the square mesh $P_n \times P_n$ it is known that $st(P_n \times P_n)  = 3n + o(n)$,
which is tight with respect to the constant factor of the largest term. This follows
from results of Schnorr and Shamir~\cite{schnorr1986optimal}, where they introduced
the $3n$-sorter for the square mesh. We also have a tight result for the general $d$-dimensional
mesh of $\Theta(dn)$ due to Kunde~\cite{kunde1987optimal}. These results are, in fact,
more general, as they apply to meshes with non-uniform aspect ratios.

\begin{table}[h]
 	\centering 
 		\begin{tabular}{c c c c} 
 			\hline\hline 
 			Graph & Lower Bound & Upper Bound & Remark \\ [0.5ex] 
 			\hline 
 			Complete Graph ($K_n$) & $\log n$ & $O(\log n)$ & AKS Network~\cite{ajtai19830}\\ 
 			Hypercube ($Q_n$) & $ \Omega(\frac{\log n \log\log n}{\log \log \log n}) $&  $2^{O(\sqrt{\log\log n})}\log n$ & Plaxton et.~al\cite{leighton1998hypercubic,plaxton1994super} \\
 			Path ($P_n$) & $n-1$ & $n$ & Odd-Even Trans.~\cite{knuth1998art} \\
 			Mesh ($P_n \times P_n$) & $3n - 2{\sqrt{n}}-3$ & $3n+O(n^{3/4})$ & Schnorr \& Shamir~\cite{schnorr1986optimal}\\
 			$d$-dimensional Mesh & $\Omega(dn)$ & $O(dn)$ &Kunde~\cite{kunde1987optimal} \\
 			{\bf Tree} &   & $ O(\min(\Delta n,n \log {(n/d)}))$& This paper \\
 			\textbf{$d$-regular Expander}& $\Omega(\log n)$ & $O(d\log^3(n))$ & This paper\\[1ex] 
 			\textbf{Complete $p$-partite} & $\Omega(\log n)$ & $O(\log n)$ & This paper \\
            \textbf{Graph} ($K_{n/p,\ldots,n/p}$)  			 \\
            {\bf Pyramid $(d, N)$}& & $O(dN^{1/d})$& This paper\\
        \hline 
    \end{tabular}
\label{table:st known results}
 	\caption{Known Bounds On The Sorting Numbers Of Various Graphs} 
\end{table}

\section{Definitions}

Formally, we study the following restricted variant of sorting networks.
We begin by taking a graph $G=(V,E)$, where the vertices correspond to the locations
of an oblivious sorting algorithm, $V=\{1,2,\ldots , |V|\}$.
The keys will be modeled by weighted pebbles, one per vertex.
Let a \textit{sorted order} of $G$ be given by a permutation $\pi$
that assigns the rank $\pi(i)$ to the vertex $i \in V$.
The edges of $G$ represent pairs of vertices where the pebbles can be compared and/or swapped.
Given a graph $G$ the goal is to design a sorting network that uses only the edges of $G$.
We formally define such a sorting network.

\begin{definition}[Sorting Network on a Graph]\label{def:sorting network}
A sorting network is a triple $\mathcal{S} (G ,\mathcal{M}, \pi)$ such that:
\begin{enumerate}
    \item $G=(V,E)$ is a connected graph with a bijection $\pi:V \to \{1, \dots, |V|\}$  specifying the sorted order on the vertices.
    Initially, each vertex of $G$ contains a pebble having some value.
	\item $\mathcal{M} = (M_i,\dots,M_t)$ is a sequence of matchings in $G$,
    for which some edges in the matching are assigned a direction.
    Sorting occurs in stages.
    At stage $i$ we use the matching $M_i \in \mathcal{M}$ to exchange the
    pebbles between matched vertices according to their orientation.
    For an edge $\overrightarrow{uv}$, when swapped the smaller of the
    two pebbles goes to $u$. If an edge is undirected then the pebbles
    swap regardless of their order.
    \item After $|\mathcal{M}|$ stages the vertex labeled $i$ contains the pebble whose rank is $\pi(i)$ in the sorted order.
        We stress that this must hold for all ($n!$) initial arrangement of the pebbles.
        $|\mathcal{M}|$ is called the \emph{depth} of the network.
\end{enumerate}
\end{definition}

\begin{definition}[Sorting Number]\label{def:sorting number}
Let $G$ be a graph, and let $\pi$ be a sorted order of $G$.
Define $st(G,\pi)$ to be the minimum depth of a sorting network $\mathcal{S}(G, \mathcal{M}, \pi)$.
The \emph{sorting number} of $G$, denoted by $st(G)$, is defined as the minimum depth of any sorting network on $G$,
i.e., $st(G) = \min_{\pi} st(G,\pi)$.
\end{definition}

\section{Our Results}

The AKS sorting network can be  trivially converted
into a network of depth $O(n \log(n))$ by making a single comparison in each round.
However, it is not clear a-priori whether for any graph there is a network
of depth $O(n \log(n))$. We show that this bound indeed holds for all graphs.

\begin{theorem}\label{thm:st bound diam}
	Let $G$ be an $n$-vertex graph, and suppose that $G$ contains a simple path of length $d$.
    Then $st(G) = O(n \log{(n / d)})$.
    In particular, for every $n$-vertex graph $G$ it holds that $st(G) = O(n \log(n))$.
\end{theorem}

This bound is tight for the star graph $K_{1,n-1}$ as at most 1 comparison can be made per round,
and hence $st(K_{1,n-1}) = \Theta(n \log(n))$.

If the maximal degree of $G$ is small, it possible to show a
better upper bound on $st(G)$.

\begin{theorem}\label{thm:st bound max deg}
    Let $G$ be an $n$-vertex graph with maximal degree $\Delta$. Then $st(G) = O(\Delta n)$.
\end{theorem}

Next, we relate the sorting number of a graph to it routing number, and the size of its maximal matching.

\begin{theorem}\label{thm:st < n log(n)  rt(G)/nu(G)}
	Let $G$ be an $n$-vertex graph with routing number $rt(G)$ and matching of size $\nu(G)$.
    Then $st(G) = O(n \log(n) \cdot \frac{rt(G)}{\nu(G)})$.
\end{theorem}

Next, we upper bound $st(G)$ for graphs $G$ that contain a large subgraph $H$ whose $st(H)$ is small.

\begin{theorem}\label{thm:st using a subgraph}
    Let $G$ be an $n$-vertex graph, and
	let $H$ be a vertex-induced subgraph of $G$ on $p$ vertices. Then
    \[
        st(G) = O\left(\frac{n}{p} \log(\frac{n}{p}) \cdot (rt(G) + st(H))\right).
    \]
\end{theorem}

Theorems~\ref{thm:st bound diam}-\ref{thm:st using a subgraph} above will be proven in Section~\ref{sec:general bounds}.

In Section~\ref{sec:concrete families} we prove bounds on some concrete families of graphs,
including the complete $p$-partite graph, expander graphs, vertex transitive graphs, Cayley graphs, and the pyramid graph.

\section{Routing via Matchings}

In order to prove some of the results in this paper we need to define the model of routing via matchings, originally introduced by Alon et.\ al~\cite{alon1994routing}.
Given a connected labeled graph $G = (V,E)$, where each vertex $i \in V$ is initially occupied by a labeled pebble that has a unique destination $\pi(i)$, the routing time $rt(G,\pi)$ is defined as the minimum number of matchings
required to move each pebble from $i$ to its destination vertex labeled $\pi(i)$,
where pebbles are swapped along matched edges. The {\it routing number} of $G$ denoted by $rt(G)$
is defined as the maximum of $rt(G,\pi)$ over all such permutations $\pi:V \to V$.
We start with the following simple lemma.
\begin{lemma}\label{lemma: st vs rt}
For any graph $G$ and any order $\pi$ of the vertices of $G$
it holds that
\[
    \max\{rt(G),\log{|G|}\} \le st(G, \pi) \le st(G) + rt(G).
\]
\end{lemma}
\begin{proof}
We show first that $rt(G,\sigma) \le st(G, \pi)$
for any two permutations $\pi, \sigma$ of the vertices.
Indeed, suppose that the keys of the pebbles are $\{1,\dots,|V|\}$.
For all $i \in V$ place the pebble ranked $i$ in the vertex $\sigma^{-1}(\pi^{-1}(i))$.
Then, there exists a sorting network of depth $st(G, \pi)$ that sends
the pebble ranked $i$ to $\pi^{-1}(i)$ for all $i \in V$.
That is, the pebble from the vertex $j = \sigma^{-1}(\pi^{-1}(i))$
is sent to the vertex $\pi^{-1}(i) = \sigma(j)$. Therefore,
$rt(G,\sigma) \le st(G, \pi)$ for all permutations $\sigma$,
and thus $rt(G) \le st(G, \pi)$. Second part of the lower bound follows from the standard information theoretic argument for oblivious sorting algorithms.
We know that any sorting network must make $\Omega(|G|\log|G|)$ compare-exchanges and the size of the largest matching is at most $|G|/2$, and hence $st(G,\pi) \geq \Omega(\log(|G|))$.

For the upper bound
let $\mathcal{S}(G, \mathcal{M}, \tau)$ be a sorting network on $G$
of depth $st(G) = st(G, \tau)$.
We use $\tau$ to create another sorting network $\mathcal{S}(H, M, \pi)$ of depth
at most $st(G) + rt(G)$. This is done in two stages.
First we apply the sorting network $\mathcal{S}(H, M, \tau)$.
After this stage we know that the pebble at vertex $i$ has a rank $\tau(i)$.
Next, we apply a routing strategy with $rt(G)$ steps that routes
to the permutation $\pi^{-1} \circ \tau$, i.e., sending a pebble
in the vertex $i$ to $\pi^{-1}(\tau(i))$ for all $i \in V$.
After this step the vertex $i$ contains the pebble of rank $\pi(i)$.
This proves that $st(G, \pi) \leq st(G) + rt(G)$.
\end{proof}

The above lemma implies that if we construct a sorting network for an
arbitrary sorted order on the vertices then we suffer a penalty of $rt(G)$
on the depth of our network as compared to the optimal one.

\subsection{Routing on subgraphs of $G$}
Below, we study the notion of routing a subset of the pebbles to a specific subgraph.
We start with the following lemma.
\begin{lemma}\label{lemma:routing to longest path}
	Let $T$ be a tree with diameter $d$, and let $P$ be a path of length $d$ in $T$.
    Then, we can route any set of $d$ pebbles to $P$ in $3d$ steps.
\end{lemma}

\begin{proof}
Denoting the number of the pebbles by $k$,
we prove that we can route \emph{any} set of $k$ pebbles to \emph{any} subset
$P_k \subseteq P$ of size $|P_k| =k$,
in at most $d+2(k-1)$ steps.
The proof is by induction on $k$
The case $k=1$ is trivial, since $d$ is the diameter of $T$.

For the induction step let $k \geq 2$, let $v_1,\dots,v_k$ be the
locations of the $k$ pebbles, and let $P_k$ be a subset of $P$ of size $k$.
Suppose without loss of generality that $dist(v_k, P_k) \geq dist(v_i, P_k)$
for all $i < k$. Note that we may assume that $dist(v_k, P_k)>0$, as otherwise all pebbles are already in $P_k$.
Let $u^* \in P_k$ be a vertex in $P_k$, that is the closest to $v_k$,
and let $(v_k=s_0,s_1,\dots,s_r = u^*)$ be the shortest path from
$v_k$ to $u^*$, where $r = dist(v_k, P) \leq d$.

By the inductive hypothesis there is a sequence of $d+2(k-2)$ matchings
routing the pebbles $v_1,\dots,v_{k-1}$ to $P_k \setminus \{u^*\}$.
We would like to argue
now that it is possible to route $v_k$ to $u^*$ using the two extra steps
by letting $v$ ``follow'' the other pebbles so as not to interrupt with
their routing from the induction hypothesis. This is indeed possible, as
we show below. Note that by applying the sequence of $d+2(k-2)$ matchings
from the induction hypothesis, after $d+2(k-2)$
rounds none of the vertices $v_1,\dots,v_{k-1}$ is
 any of the vertices $\{ s_{r-2},s_{r-1},s_r=u^* \}$.
This is because $(v_k=s_0,s_1,\dots,s_r = u^*)$ is the shortest path from
$v_k$ to $P_k$, and in particular $s_{r-2},s_{r-1} \notin P_k$.
By the assumption that $dist(v_k, P_k) \geq dist(v_i, P_k)$,
it follows that after the $d+2(k-2)-1$ rounds (i.e., one round before the last one)
none of the vertices in $v_1,\dots,v_{k-1}$ is in $\{s_{r-3},s_{r-2}\}$.
This is because $dist(s_{r-2},P_k) = 2$ and $dist(s_{r-3},P_k) = 3$, and so in the last
step a vertex from $\{s_{r-3},s_{r-2}\}$ could not reach $P_k$.
Analogously, for all $j \leq r$ after the $(d+2(k-2)-j)$ rounds
none of the vertices in $v_1,\dots,v_{k-1}$ is in $\{s_{r-j-2},s_{r-j-1}\}$.
Therefore, (recall that $r \leq d$) we may take the routing sequence from the inductive
hypothesis, and augment its last $r-2$ steps by moving the pebble from $v_k$ to $u^*$
along the shortest path $(v_k=s_0,s_1,\dots,s_{r-2})$, and then in the last 2 rounds
the matchings will be just singletons moving this pebble from $s_{r-2}$ to $u^*$.
\end{proof}

Motivated by the lemma above, we discuss below the question of partial routing,
where only a small number of pebbles are required to reach their destination.
\begin{definition}
	Given a graph $G=(V,E)$ let $A,B \subset V$ be two subset of vertices with $|A|=|B|$,
    not necessarily distinct. Let $\pi_{AB}$ be a bijection between $A$ and $B$.
    Routing of the pebbles from $A$ to their respective destinations on $B$
    given by $\pi_{AB}$ is defined as a partial routing in $G$,
    where each pebble in $a \in A$ in required to reach $\pi_{AB}(a) \in B$
    using the edges of $G$
    (and there are no requirements on the pebbles outside $A$).
    \begin{enumerate}
      \item Define $rt(G,A,B,\pi_{AB})$ be the minimum number of matchings
      needed to route every pebble $a \in A$ to $\pi_{AB}(a) \in B$ using the edges of $G$.
      \item Define $rt(G,A,B) = \max_{\pi_{AB}} rt(G,A,B,\pi_{AB})$.
      \item For $U \subseteq V$ define $rt_U(G) = \max_{A \subseteq V} rt(G,A,U)$.
      \item For $p \in 1,\dots, |V|$ define $rt_{p}(G) = \max_{A,B \subset V, |A| = |B| \le p}{rt(G,A,B)}$.
    \end{enumerate}
\end{definition}

Clearly, for any $n$-vertex graph $G$ we have $rt(G) = rt_n(G)$.
Some of the bounds for $rt(G)$ also holds for $rt_p(G)$.
For example, $rt_p(G) \ge d$, where $d$ is the diameter of $G$.
Furthermore, $rt_p(G) = \Theta(rt(G))$ for any $p$ if and only if
$rt(G) = \Theta(d)$. We illustrate $rt_p(G)$ by computing it
explicitly for some typical graphs. Recall that $rt(K_n) = 2$.
It is easy to see that $rt_p(K_n) = 2$ for all $p \ge 3$, and $rt_2(K_n) = 1$.
For the complete bipartite graph we have $rt_{n/2}(K_{n/2,n/2}) = 2$ and is $rt_{p}(K_{n/2,n/2}) = 4$ for $p > n/2$.

\begin{theorem}
	For any tree $T$ with diameter $d$, $rt_{p}(G) = O((d+p) \min(d, \log {n \over d}))$.
\end{theorem}
\begin{proof}
The proof is similar to the proof used in~\cite{alon1994routing} for determining the routing number of trees. We find a vertex $r$ whose removal disconnects the tree into a forest of trees each of which is of size at most $n/2$. Let $(T_1,\ldots,T_r)$ be the set of trees in the forest, with $r \in T_1$.  For a tree $T_i$ let $S_i$ be the set of ``improper'' pebbles that need to be moved out of $T_i$. All other pebbles in $T_i$ are ``proper''.
In the first round we move all the pebbles in $S_i$ as close to the root of $T_i$ as possible, for all $i$. Using the argument used in~\cite{alon1994routing} it can be shown that for a tree with diameter $d$ this first phase can be accomplished in $2d$ steps for some constant $c_1$.
First we label each node in $T_i$ as odd or even based on their distance from $r_i$, the root of $T_i$.
In each odd round we match nodes in odd layers with proper pebbles to one of its children containing an improper pebble if one exists.
Similarly, in even rounds we match nodes in even layers with proper pebbles to one of its children containing an improper pebble if one exists.
Since $T$ has diameter $d$ any path from $r_i$ to some leaf must be of length at most $d-1$.
Now consider an improper pebble $u$ initially at distance $k$ from the root. During a pair of odd-even matchings either the pebble moves one step closer to the root or one of the following must be true: (1) another pebble from one of its sibling node jumps in front of it or (2) there is some improper pebble already in front of it. It can then be argued (we omit the details here) that after $c_1d$ matchings for some constant $c_1$  if $u$ ends up in position $j$ from $r_i$ then all pebbles between $u$ and $r_i$ must be improper.  Next we exchange a pair of pebbles between subtrees using the root vertex $r$, since at most $p/2$ pairs needs to be exchanged, the arguments used in~\cite{alon1994routing} can be modified to show that this phase also takes $c_2p$ steps for some constant $c_2$. After each pebble is moved to their corresponding destination subtrees we can route them in parallel. Noting that each tree $T_i$ has diameter at most $d-1$. Hence we have the following recurrence:
	\begin{align}
		T(n, d, p) &\le T(n/2,d-1,p) + c_1d+c_2p
	\end{align}
where $T(n,d,p)$ is the time it takes to route $p$ pebbles in a tree of diameter $d$ with $n$ vertices.
Taking $T(\cdot, d, p) = O(d)$, and solving (1) gives the stated bound of the lemma.
\end{proof}

\section{General Upper Bounds on $st(G)$}\label{sec:general bounds}

Below we prove Theorem~\ref{thm:st bound diam}
stating that if $G$ contains a simple path of length $d$, then
$st(G) = O(n \log{(n / d)})$.

\begin{proof}[Proof of Theorem~\ref{thm:st bound diam}]
    It is easy to see that if $G$ contains a simple path of length $d$,
    then $G$ has a spanning tree $T$ such that with diameter at least $d$.
	The proof follows easily from Theorem~\ref{thm:st using a subgraph}
    and Lemma~\ref{lemma:routing to longest path}.
    Indeed, in the setting of Theorem~\ref{thm:st using a subgraph}, let $H$ be
    a path of length $d$ in $T$. Then  $st(H) = d$.
    By  Theorem~\ref{thm:st using a subgraph}
    if any set of $d$ vertices can be routed to $H$ in $r$ steps,
    then $st(T) \leq O(\frac{n}{d} \log( \frac{n}{d}) \cdot (r + st(H)))$.
    By Lemma~\ref{lemma:routing to longest path} we have $r= O(d)$,
    and thus $st(G) \leq st(T) = O(n \log( n/d ))$.
\end{proof}

Next, we prove Theorem~\ref{thm:st bound max deg}.
The proof is essentially from~\cite{AS}, who proved that the acquaintance time of a $G$,
defined in~\cite{benjamini2014acquaintance}, is upper bounded by $20 \Delta n$.
The basic idea is to use an $n$ round sorting network
for $P_n$, and simulate this network in $T$ with an overhead that depends only on $\Delta$.

\begin{proof}[Proof of Theorem~\ref{thm:st bound max deg}]
    Clearly it is sufficient to prove that for a spanning subgraph $T$ of $G$
    it holds that $st(T) \leq 20 \Delta n$.
    A \emph{contour} of a tree $T$ is a closed walk on $2n-1$ vertices
    that crosses each edge exactly twice, and visits each vertex $v$ exactly
    $\deg(v)$ times.  Such a contour can be constructed by considering a
    depth-first search walk on $T$.

    Let $\Gamma$ be a contour of $T$. Consider $\Gamma$ as a path on $2n-1$ vertices,
    and let $\pi$ denote the projection from $\Gamma$ to $T$.
    We claim that for each $x \in T$
    it is possible to choose a vertex of $\Gamma$ from $\pi^{-1}(x)$ so that
    the gaps between the consecutive chosen vertices along $\Gamma$ are
    at most $3$. To do this, let $r$ be an arbitrarily chosen first vertex of $\Gamma$.
    Recall that $T$ is assumed to be a tree.  For each vertex $x \in T$
    whose distance from $r$ is even we pick the first vertex of $\Gamma$ projecting to
    $x$, and for each vertex $x \in T$ whose distance from $r$ is odd we pick the last
    one. Note that $\Gamma$
    visits each leaf of the tree exactly once. Between consecutive visits to the
    leaves the contour descends towards the root, and then ascends to the
    next leaf.  Along the descent, the vertices are visited for the last
    time, and so every other vertex is selected. Along the ascent, the
    vertices are visited for the first time, and also every other vertex is
    selected. Hence, we have at most three steps of
    $\Gamma$ between any two consecutive selected vertices.

    Consider the odd-even transposition sort in $P_n$,
    the $n$-vertex path whose vertices are denoted by $\{1,2,\dots,n\}$.
    In the odd rounds we compare the pebbles on the edges $\{(i,i+1) : i \mbox{ odd}\}$.
    In the even rounds we compare the pebbles on the edges $\{(i,i+1) : i \mbox{ even}\}$.
    It is known that after $n$ rounds the pebbles in $P_n$ are sorted~\cite{knuth1998art}.

    In order to present a sorting network in $T$ with at most $20 \Delta \cdot n$ rounds
    we emulate the foregoing sorting network in $P_n$
    by simulating each round of the transposition sort with a
    sequence of at most $20 \Delta$ rounds in $T$.

    First, consider the path $\Gamma$ with $2n-1$ vertices, and place $n$
    pebbles on $\Gamma$ in the selected vertices of $\Gamma$, so that
    the distance between any two consecutive pebbles is at most $3$.
    Our goal is to sort the $n$ pebbles on $\Gamma$,
    which, in particular, implies sorting in $T$.
    Since the pebbles are not located in consecutive vertices of the path,
    every round of the odd-even transposition sort for $P_n$ will require
    several (at most 5) rounds of sorting in $\Gamma$. We then show
    how to emulate these moves by a sorting network in $T$ with the sorted order
    defined by the linear order of the marked vertices of $\Gamma$.

    Let $i < j$ be two consecutive marked vertices of $\Gamma$, and let
    $p_i$ and $p_j$ be the pebbles in the corresponding vertices.
    In order to compare the pebbles $p_i$ and $p_j$ we can perform a sequence of
    swaps (without comparisons) in $\Gamma$ along the edges $(i,i+1), \dots, (j-2,j-1)$, which brings
    the pebble $p_i$ to the vertex $j-1$, followed a comparison step $(j-1,j)$,
    which places the pebble $\max(p_i,p_j)$ into the vertex $j$,
    and finally, performing the sequence of
    swaps $(j-1,j-2), \dots, (i+1,i)$ (again, without comparisons)
    in order to bring the pebble $\min(p_i,p_j)$ to the vertex $i$.
    The gaps between consecutive pebbles are at most $3$, and hence it takes at most
    $5$ steps on $\Gamma$ to perform such a swap.  These swaps projected on
    $T$ result in comparison of the pebbles in the vertices $\pi(i)$ and
    $\pi(j)$, leaving all other pebbles in their place.

    The difficulty is that in the graph $T$ the steps for swapping a pair of
    pebbles $p_i$ and $p_j$ could interfere with the steps for swapping
    another pair $p_{i'}$ and $p_{j'}$.  This happens if the projections to
    $T$ of the intervals $[i,j]$ and $[i',j']$ in the path $\Gamma$
    intersect.  If there are no such intersections, then all the swaps for all pairs
    could be carried out in parallel and we would have a sorting network in $T$ with $5n$ rounds.

    In order to solve this problem, we separate each round of
    odd-even transposition sort in $P_n$ into several sub-rounds,
    so that conflicting pairs of intervals are in different sub-rounds,
    and then split each sub-round into at most $5$ steps in $\Gamma$
    as described above.
    By the assumption on the maximal degree on $T$, the path
    $\Gamma$ visits each vertex of $T$ at most $\Delta$ times. Thus, since the
    intervals $[i,j]$ of $\Gamma$ that we care about in each round are vertex disjoint in
    $\Gamma$, each vertex of $T$ is contained in at most $\Delta$ such intervals.
    Each interval consists of at most $4$ vertices of $T$, and therefore each interval
    $[i,j]$ is in conflict (i.e., their projections to $T$ intersect) with at most
    $4(\Delta-1)$ other intervals $[i',j']$ participating in the current round.

    It is well known that if a graph has maximum degree $D$, then its
    chromatic number is at most $D+1$.  Applying this to the conflict graph
    of the intervals to be swapped in one round of the odd-even transposition sort
    we see that we can assign each interval $[i,j]$ of $\Gamma$ one of $4\Delta-3$
    colors, so that conflicting intervals have different colors.

    We now split each round of the transposition sort on $\Gamma$ into $4\Delta-3$
    sub-rounds, where in every sub-rounds we swap all pairs of the same color
    that are to be swapped in this round of the transposition sort in $P_n$.  Finally,
    we simulate this strategy in $P_n$ by replacing each sub-round with at
    most $5$ steps in $\Gamma$ as described above. Our coloring of the
    intervals guarantees that there are no conflicting intervals,
    and hence the comparisons can be carried out in $T$ in parallel.
    Therefore, each round of the odd-even transposition sort can
    be simulated by $5(4\Delta-3)$ rounds in $T$, and hence
    $st(T, \pi) \leq 5(4\Delta - 3) n$, where $\pi$ is the
    order defined by the natural linear order in $\Gamma$.
    This completes the proof of the theorem.
\end{proof}

Next, we prove Theorem~\ref{thm:st < n log(n)  rt(G)/nu(G)}
stating that $st(G) = O( n \log(n) \cdot \frac{rt(G)}{\nu(G)})$,
where $rt(G)$ is the routing number of $G$, and $\nu(G)$ is the size of the maximal matching in $G$.
\begin{proof}[Proof of Theorem~\ref{thm:st < n log(n)  rt(G)/nu(G)}]
	We prove the theorem by using $G$ to simulate the AKS sorting network
    on the complete graph $K_n$ of depth $O(\log(n))$.
    Specifically, we show that each stage (a matching) of the sorting network on $K_n$
    can be simulated by at most $O(\frac{n}{\nu(G)} rt(G))$ stages (matchings) in $G$.
    Let $M$ be a matching at some stage of the AKS sorting network on the complete graph.
    We simulate the compare-exchanges and swaps in $M$ by a sequence of matchings in $G$ as follows.
    First we partition the edges in $M$ into $t = \lceil n / \nu(G)\rceil$ disjoint subsets
    $M = M_1 \cup \dots \cup M_t$, where $|M_i| = \nu(G)$ for all except maybe the last set $M_t$, which can be smaller. Let $M_G$ be a maximal matching in $G$. Corresponding to each pair $(u,v) \in M_i$ we pick a distinct pair $(u',v') \in M_G$, this can always be done since the sets $M_i$ and $M_G$ are of the same size. Note that the pair $(u,v)$ may not be adjacent in $G$, and so, we route each pair $(u,v) \in M_i$ to its destination in $(u',v') \in M_G$. This can be done in $rt(G)$ steps, where each step consists of only undirected matchings. Once the pairs have been placed in to their corresponding positions we relabel the vertices such that the pair labeled $(u',v')$ is now $(u, v)$. Unmatched vertices keep their label. Since the pairs in $M_i$ are now adjacent in $G$ we can perform the compare-exchange or swap operation according to $M_i$.
    Therefore, the total number of matchings to execute the $i^{th}$ set of compare-exchanges and swaps in $M_i$ is $rt(G) + 1$ in $G$.
    We remark that the routing maintains the oblivious nature of the network,
    and the swaps are made while routing, which are data independent.
    We perform all the operations in each $M_i$ successively,
    while keeping track of the relabeling of the vertices that occur at each sub-stage.
    This implies that we can simulate $M$ using at most
    $(rt(G)+1) \cdot t = O(\frac{n}{\nu(G)} \cdot rt(G))$ matchings in $G$. Therefore, since the depth of the AKS sorting network
    on the complete graph $K_n$ is $O(\log(n))$, we conclude that
    $st(G) = O( n \log(n) \cdot \frac{rt(G)}{\nu(G)})$, as required.
\end{proof}

Next we prove Theorem~\ref{thm:st using a subgraph}, saying that
if $H$ be a vertex-induced subgraph of $G$ on $p$ vertices then
$st(G) = O\left(\frac{n}{p} \log(\frac{n}{p}) \cdot (rt(G) + st(H))\right)$.

\begin{proof}[Proof of Theorem~\ref{thm:st using a subgraph} ]
	Let us partition the vertex set $V$ of $G$ into $q = \lceil n/\lfloor p/2\rfloor\rceil$
    parts $V = A_1 \cup \dots \cup A_q$ in a balanced manner (i.e., the size
    of each $A_i$ is either $\lfloor p/2 \rfloor$ or $\lfloor p/2 \rfloor - 1$).
    Let $K_q$ be a complete graph whose vertices are identified with  $\{A_1,\ldots,A_q\}$,
    and let $S$ be an oblivious sorting algorithm with $O(q\log q)$ comparisons
    on the complete graph $K_q$.
    (Here the sequence of comparisons is performed sequentially, and not in parallel.)
    In an ordinary sorting network in each step
    we perform a compare-exchange or a swap between two matched vertices $(i,j)$
    so that if $i < j$, then the pebble in the vertex $i$ will be smaller than the pebble in $j$
    We will simulate $S$ on $G$ using a sorting network on $H$ by sorting in each stage the elements in $A_i \cup A_j$.
    That is, for $i<j$ we are going to sort the elements in $A_i \cup A_j$
    so that all the elements of $A_i$ are smaller than every element of $A_j$,
    and the elements within each subset are internally sorted. This is done
    using an optimal sorting network in $H$, which we will denote by $\mathcal{S}_H$.
	
    The key observation is that we can simulate any such compare-exchange in $G$
    between pairs of sets in $A$ in $O(rt(G) + st(H))$ steps.
    Indeed, suppose the $k^{th}$ round in $S$ compares the vertices $i<j$.
    In order to simulate this comparison we first route all the pebbles in $A_i \cup A_j$
    to the subgraph $H$ and relabel the vertices.
    This relabeling is done so that we can keep track of the vertices when sorting $H$.
    Then, we use $\mathcal{S}_H$ to sort $A_i \cup A_j$ which takes $st(H)$ steps.
    Note that if $|A_i \cup A_j| < p$ we can still use the network $S_H$ to sort it by slightly modifying the original network. Once the sorting is done we split up the sets again and appropriately relabel the vertices so that the first $|A_i|$ vertices in the sorted order on $H$ will now belong to $A_i$ and the next $|A_j|$ vertices will belong to $A_j$. If instead the $k^{th}$ comparison is actually a swap then we simply swap the labels of the multisets ($A_i$ is labeled $A_j$ and vice versa). Hence performing the above simulation takes $O(rt(G) + st(H))$ steps per compare exchange or swap operation, which gives the result of the theorem.
\end{proof}

In the proof of Theorem~\ref{thm:st using a subgraph} above we only used an oblivious sorting algorithm with $O(q\log q)$ comparisons on the complete graph $K_q$, and did not use the fact that the comparisons can be done
in parallel, e.g., using the AKS sorting network.
This is because Theorem~\ref{thm:st using a subgraph} only assumes that there is one subgraph $H$ with
small $st(H)$. If instead we assumed that there are many such subgraphs, then we could sort the $A_i$'s
in different subgraphs in parallel. This is described in the corollary below.

\begin{corollary}\label{cor:st using many subgraphs}
    Let $G=(V,E)$ be an $n$-vertex graph.
    Let $V = V_1 \cup \dots V_q$ be a partition of the vertices,
    with $|V_i| = n/q$ for all $i \in \{1, \dots, q\}$,
    and let $H_i$ be the connected subgraph induced by $V_i$ for each $i \in \{1, \dots, q\}$.
	Then
    \[
        st(G) = O \left( \log(q) \cdot (rt(G) + \max_{k \in \{1, \dots, q\}}\{st(H_k)\})\right).
    \]
\end{corollary}
\begin{proof}[Proof sketch]
    The proof uses the same idea that Theorem~\ref{thm:st using a subgraph}.
	We start by partitioning the vertex set $V$ of $G$ into $2q$ parts
    $V = A_1 \cup \dots \cup A_{2q}$ of equal sizes.
    Then, we simulate oblivious sorting algorithm on $K_{2q}$ with the sets $A_i$.
    The only difference is that instead of an oblivious sorting algorithm with $O(q\log(q))$ comparisons
    on the complete graph $K_{2q}$ we use the AKS sorting network on $2q$ vertices of depth $O(\log(q))$.
    In each round of the sorting network there are at most $q$ comparisons,
    and the corresponding sorting of $A_i \cup A_j$ can be performed in parallel, one in each $H_k$
    in time $st(H_k)$.
\end{proof}

\section{Bounds on Concrete Graph Families}\label{sec:concrete families}

Below we state several results concerning the sorting time of some concrete families of graphs.

\begin{proposition}[Complete $p$-partite graph]\label{prop:st(complete p partite graph)}
	Let $G$ be the complete $p$-partite graph $K_{n/p,\ldots,n/p}$ on $n$ vertices.
    Then $st(G) = \Theta(\log n)$.
\end{proposition}

\begin{proof}
The lower bound is trivial.
For the upper bound
note that $K_{n/p,\ldots,n/p}$
contains the bipartite graph $K_{\floor{\frac{p}{2}}\frac{n}{p}, \ceil{\frac{p}{2}}\frac{n}{p}}$.
In particular, it contains a matching of size $\nu(G) = \floor{\frac{p}{2}} \cdot n/p$ since $K_{n/p,\ldots,n/p}$.
Therefore, by Theorem~3 in~\cite{alon1994routing} and the remark after the proof,
we have $rt(G) \leq rt(K_{\floor{\frac{p}{2}}\frac{n}{p}, \ceil{\frac{p}{2}}\frac{n}{p}})
\leq 2 \ceil{ \frac{\ceil{\frac{p}{2}}}{\floor{\frac{p}{2}}} }+2 \leq 6$, and hence
by Theorem~\ref{thm:st < n log(n)  rt(G)/nu(G)}
it follows that $st(G) \leq O(\log(n))$.
\end{proof}

Recall that a graph $G$ is said to be a $(n,d,\lambda)$-expander
if it is a $d$-regular graph on $n$ vertices
and the absolute value of every eigenvalue of its adjacency matrix
other than the trivial one is at most $\lambda$.
\begin{proposition}[Expander graphs]\label{prop:st(expander)}
	Let $G$ be an $(n,d,\lambda)$-expander. Then $st(G) \leq O(\frac{d^3}{(d-\lambda)^2}\log^3(n) )$.
    In particular, if $\lambda < (1-\frac{1}{\log^c(n)})d$, then
    $st(G) \leq O(d \cdot \log^{2c+3}(n) )$.
\end{proposition}
\begin{proof}
	Recall from~\cite{alon1994routing} that if
    $G$ is an $(n,d,\lambda)$-expander, then
    $rt(G) = O\left( \frac{d^2}{(d-\lambda)^2}\log^2(n) \right)$.
    Therefore, since any $d$-regular graph contains a matching of size $n/2d$
    it follows from Theorem~\ref{thm:st < n log(n)  rt(G)/nu(G)}
    that $st(G) \leq O(\frac{d^3}{(d-\lambda)^2}\log^3(n) )$.
\end{proof}

\begin{proposition}[Vertex transitive graphs]\label{prop:st(vertex transitive)}
	Let G be a verter transitive graph with $n$ vertices of degree $\polylog(n)$. Then $diam(G) = O(\polylog(n))$ if and only if $st(G) =  O(\polylog(n))$.
\end{proposition}
\begin{proof}
It is trivial that $diam(G) \leq st(G)$.
For the other direcion, Babai and Szegedy~\cite{babai1992local} showed that for vertex-transitive
graphs if the diameter of $G$ is $O(\polylog(n))$ then its vertex expansion is
$\Omega(1/\polylog(n))$. Therefore, $\lambda \le d(1 - 1/\polylog(n))$, where $d = O(\polylog(n))$ is the degree of the graph.
Therefore, by Proposition~\ref{prop:st(expander)} we have $st(G) = O(\polylog(n))$.
\end{proof}

Next we bound the sorting number of cartesian product of two given graphs. Recall that for two graphs $G_1(V_1,E_1)$ and $G_2(V_2,E_2)$ their Cartesian product is $G_1 \square G_1$
is the graph whose set of vertices is $V_1 \times V_2$ and
$((u_1,u_2),(v_1,v_2))$ is an edge in $G_1 \square G_1$ if either
$u_1 = v_1$, $(u_2,v_2) \in E_2$ or $(u_1,v_1) \in E_1$, $u_2 = v_2$.
Our next result bounds the sorting number of a product graph in terms of sorting numbers of its components.

\begin{figure}[h]
	\includegraphics[width=6.5cm]{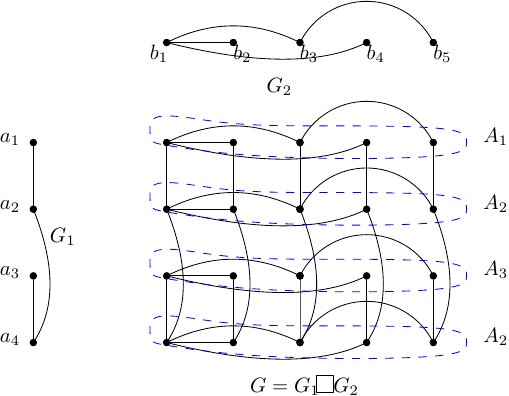}
	\centering
	\caption{The product graph $G = G_1 \square G_2$. The rows highlighted by blue regions represents copies $G_2$. }
\label{fig:product}
\end{figure}

\begin{corollary}\label{cor:st graph product}
    Let $G_1 = (V_1,E_1)$ and $G_2 = (V_2,E_2)$ be two graphs and $G = G_1 \square G_2$.
    Then
	\[
        st(G) \leq
        O(\min(\log|V_1|(rt(G)+ st(G_2)), \log|V_2|(rt(G) + st(G_1))).
    \]
\end{corollary}

\begin{proof}
We will prove the corollary in terms of $rt(G)$, and then use Theorem~4 in~\cite{alon1994routing}
saying that $rt(G) \le \min\{rt(G_1),rt(G_2)\} + rt(G_1) + rt(G_2)$.
Since $G$ has $|V_1|$ vertex disjoint subgraphs that are copies of $G_2$ we can apply
Corollary~\ref{cor:st using many subgraphs} with these $q=|V_1|$ subgraphs,
and all $H_i$ being isomorphic to $G_2$.
Therefore, we get
\[
    st(G) \leq O(\log(|V_1|) \cdot (rt(G) +st(G_2) )
\]
The bound $st(G) \leq O(\log(|V_2|) \cdot (rt(G) +st(G_1) )$
follows using the same argument by changing the roles of $G_1$ and $G_2$.
\end{proof}

As an example of an application of the above corollary consider the $d$-dimensional mesh $M_{n,d}$ with $n^d$ vertices. We know that $rt(M_{n,d}) \le 2dn$ since $M_{n,d} = M_{n,d-1} \times P_n$.
Therefore, $st(M_{n,d}) \le O( \log(n^{d-1}) \cdot (rt(M_{n,d})  + st(P_n)) ) = O(d n \log(n))$.
Although this bound is not optimal (it is known that $st(M_{n,d}) = O(dn)$),
we still find this example interesting.

\subsection{The Pyramid Graph}
A 1-dimensional pyramid with $m$-levels is defined as the complete binary tree of $2^{m} -1$ nodes,
where the nodes in each level are connected by a path (i.e., a one-dimensional mesh).
We treat the apex (root) to be at level 0, and subsequent levels are numbered in ascending order.
A 2-dimensional  pyramid is shown in Figure~\ref{fig:pyramid}. In this case each level $l$ is a $2^l \times 2^l$ square mesh.
Similarly a $d$-dimensional pyramid having $m$ levels is denoted by $\triangle_{m,d}$, where the level $l$ is a $d$-dimensional
regular mesh of length $2^{l}$ in each dimension.
Clearly, the size of layer $l$ is $|M_l| = 2^{ld}$
and the number of vertices in the graph is $N = |\triangle_{m,d}|= \sum_{l=0}^{m-1} 2^{ld}= \frac{2^{md}-1}{2^{d}-1}$.
We treat a vertex $x \in M_l$ as a vector in $[1,2^l]^{d}$ which denotes its position on the mesh.

\begin{figure}[h]
	\includegraphics[width=4.5cm]{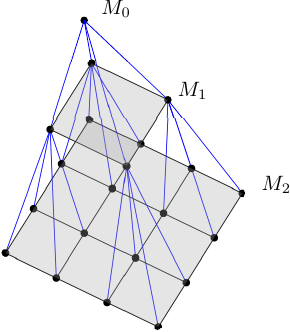}
	\centering
	\caption{A pyramid $\triangle_{3,2}$ in 3-dimension}
\label{fig:pyramid}
\end{figure}

In this section we prove an upper bound on $st(\triangle_{m,d})$. In order to derive this bound we make use of the following bound on the routing number of pyramid.

\begin{lemma}\label{lemma: rt pyramid}
	Let $\triangle_{m,d}$ be the $d$-dimensional pyramid graph with $m$-levels. Then $rt(\triangle_{m, d}) = O(dN^{1/d})$.
\end{lemma}

\begin{proof}
	\begin{figure}[h]
		\includegraphics[width=4.5cm]{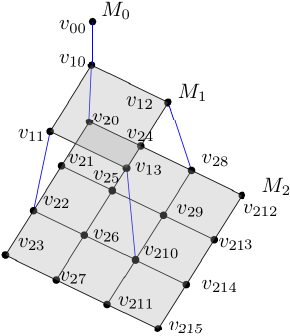}
		\centering
		\caption{The graph $\triangle'_{3,2}$ after stripping way edges from $\triangle_{3,2}$}
    \label{fig:subpyramid}
	\end{figure}
	
Given the pyramid $\triangle_{3,2}$ consider a subgraph $\triangle'_{3,2}$ as shown in Figure~\ref{fig:subpyramid}. In literature this graph is sometimes refer to as a multi-grid, see for example~\cite{leighton2014introduction}.  As we move down from the apex we remove all but the ``first'' edge from the set of edges that connects a vertex to its neighbors in the level below. The remaining edges that connects two adjacent layers will be referred to as vertical edges. These edges can be grouped into disjoint vertical paths as shown by the blue lines in Figure~\ref{fig:subpyramid}.  The above construction naturally generalizes in higher dimensions. Clearly $rt(\triangle_{m,d}) \le rt(\triangle'_{m,d})$ where $ \triangle'_{m,d} $ is the multi-grid obtained from $ \triangle_{m,d} $.  We shall show $rt(\triangle'_{m,d}) = O(dN^{1/d})$.
	
	Let $\pi$ be some input permutation. Without loss of generality we assume that $\pi$ consists only of 2-cycles or 1-cycles. From~\cite{alon1994routing} we know that any arbitrary permutation can be written as a composition of at most two such permutations. In order to route $\pi$ we first route the pebbles into their appropriate levels and then route within these levels. Routing consists of five rounds where in the odd numbered rounds we route within the levels and in the even numbered rounds we use the vertical paths to route between the levels. The first four rounds are used to move the pebbles to their appropriate destination level.
	
	Let $v_{ij}$ be the $j^{th}$ node at level $i$, where $j \in [0, n_i-1]$. Let $\phi_k$ be the number of maximal vertical paths of length $k$. For example, in Figure~\ref{fig:subpyramid} we have $\phi_2=1$ and $\phi_1=3$. In general in a $\triangle'_{m,d}$, $\phi_k = n_{m-k-1}-n_{m-k-2}$ for $k \in [1,m-2]$ and $\phi_{m-1}=1$. We group the cycles in $\pi$ based on their source and destination level (in case of 1-cycles the source and destination levels are the same). Let $P_{ij}$ ($i < j$) be the set of pebble pairs that need to be moved from level $i$ down to level $j$ and vice-versa and $P_{ii}$ be the set of pebbles that stay in level $i$. Let $\mu_{ij} = |P_{ij}|$. Let $P_i=\bigcup_{i<j}P_{ij}$ be the set of pebble pairs that move a pebble up to level $i$.
	We shall only use disjoint vertical paths of length $m-i-1$ to route the pebbles in $P_{i}$.
	During either even round each path of length $m-i-1$ will be used, for some $j$, to swap two pebbles between levels $i$ and $j$; all other pebbles on that path will not move. As an example consider the case in Figure~\ref{fig:subpyramid}. Suppose $\pi(v_{00}) = v_{21}$. Then during the intra-level routing on the first round we will move the pebble at $v_{21}$ to $v_{20}$. All intermediate nodes on this path, which in this case is just $v_{10}$ will be ignored (i.e., a pebble on these nodes will return to their original position at the end of the round). The four pebbles $\{v_{10},v_{11},v_{12},v_{13}\}$ will only use the three paths of length 1 to move to the bottom level (if necessary). In general $|P_i|=\sum_{j > i}{\mu_{ij}} \le n_i \le 2(n_{i}-n_{i-1}) = 2\phi_{m-i-1}$. Hence we need at most two rounds of routing along these vertical paths to move all pebbles in $P_i$.
	
	Routing within the levels (which happens in parallel) is dominated by the routing number of the last level which is known to be $O(dn_{m}^{1/d}) = O(dN^{1/d})$ (for example, we can use Corollary 2 of Theorem 4 in~\cite{alon1994routing}). Hence the three odd rounds take $O(dN^{1/d})$ in total. In the even rounds routing happens in parallel along the disjoint vertical paths. The routing time in this case is $O(m)$. Since, $N^{1/d} = \Omega(2^m)$, the even rounds do not contribute to the overall routing time, which remains $O(dN^{1/d})$, as claimed by the theorem.
	
\end{proof}

Using the above theorem we give an upper bound on the sorting number of the pyramid.

\begin{theorem}[pyramid]\label{thm: st pyramid}
	The sorting number for a pyramid $\triangle_{m,d}$ is $O(d\  N^{1/d})$.
\end{theorem}
\begin{proof}
	Let $\triangle_{i,d}$ denote the sub-pyramid from level 0 to $i$ and let $M_{i}$ be the $d$-dimensional mesh at level $i$. Let $\pi_i$ be the sorted order on the mesh $M_i$. Note that $\pi_i: [1,2^i]^{d} \to [n_i]$ is a bijection and $\pi_0$ is the identity permutation of order 1. Next we define a sorted order $\pi$ for the pyramid $\triangle_{m,d}$ based on the $\pi_i$'s. In $\pi$ we assume the layers are sorted among themselves in ascending order starting from the apex. So the vertex labeled (with respect to $M_i$) $i$ on layer $j$ has a global rank $\pi(i) = \pi_j(i) + |\triangle_{j-1,d}|$.  Recall that $st(M_i) = O(dn_i^{1/d})$ which is due to Kunde\cite{kunde1987optimal} where he used the general snake-like ordering. From Lemma~\ref{lemma: st vs rt} we see that this bound still holds if we replace the snake-like ordering with some arbitrary permutation. Obviously in this case $rt(M_{i}) = \Theta(st(M_{i}))$.  Next we describe the matchings $\mathcal{M}$ of sorting the network $\mathcal{S}(\triangle_{m,d},\mathcal{M},\pi)$ in terms of an oblivious sorting algorithm described below.

    \medskip
	\begin{enumerate}
    \item[] $\mathcal{S}(\triangle_{m,d},\mathcal{M},\pi)$	
    \item Route all pebbles of $\triangle_{m-1,d}$ to $M_{m-1}$ and sort them using the mesh.
	\item Route these pebbles back to $\triangle_{m-1,d}$ such that they are in sorted order (according to $\pi$).
	\item Sort the mesh $M_{m-1}$ according to $\pi_{m-1}$.
    \item Route a pebble of rank $i \le n_{m-2}$ at position $x_i \in M_{m-1}$ to $y_i \in M_{m-1}$ where $$y_i[j] = 2\pi_{m-2}^{-1}{(n_{m-2}+1-i)}[j]-1$$ Let $Y = (y_1,\ldots,y_{n_{m-2}})$.
    \item Merge $Y$ with $M_{m-2}$ using pair-wise compare-exchanges, where $y_i$ is compared with $z \in M_{m-2}$ such that $\pi_{m-2}(z) = i$.
    \item Repeat 1-5.
    \item Repeat 1-3.

	\end{enumerate}

\paragraph*{Running time}
Note that the number of times we route on $\triangle_{m,d}$ is 6. Also sorting on the mesh $M_{m-1}$ occurs 6 times. We know that both routing and sorting on a mesh takes $O(dN^{1/d})$ steps. From Lemma~\ref{lemma: rt pyramid} we see that routing on $\triangle_{m,d}$ also takes $O(dN^{1/d})$ steps. So the total contribution of all the steps except 4 and 5 is $O(dN^{1/d})$. It is easy to see that step 4 also takes $O(dN^{1/d})$ and the step 5 can be accomplished in constant time. Putting it all together we see that $st(\triangle_{m,d}) = O(dN^{1/d})$ as claimed.

\paragraph*{Proof of correctness} Next we give proof sketch that $\mathcal{S}(\triangle_{m,d},\mathcal{M},\pi)$ is a sorting network. Clearly the algorithm is oblivious, hence we invoke the 0-1 principle~\cite{knuth1998art} and assume that our pebbles are all 1's and 0's.  Before the  execution of step 7 if every pebble in $\triangle_{m-1,d}$ is smaller than every pebble in $M_{m-1}$ then after step 7 we shall have our desired sorted order. So lets assume to the contrary it is not. Then there must be some pebbles $x \in M_{m-1}$ that suppose to be in $\triangle_{m-1,d}$. If that is the case then $x$ must be a 0 otherwise $x$ is $\ge$ every pebble in $\triangle_{m-1,d}$ and we are done. Now let us look at step 4 and 5. In step 4 we route the set of $n_{m-2}$ smallest pebbles in $M_{m-1}$ such that the $i^{th}$ smallest pebble is at some vertex of $M_{m-1}$ which  is directly connected to the vertex in $M_{m-2}$ that has the $i^{th}$ largest pebble of $M_{m-2}$. Since $x$ was not exchanged during both the iteration of step 4 and 5 then $x$ must be larger than at least $2n_{m-2}$ elements in $\triangle_{m,d}$, but then $x$ should not belong to $\triangle_{m-1,d}$ (since $|\triangle_{m-1,d}| \le 2n_{m-2} - 1$ for any $d$) contradicting our assumption. Hence, after step 6 we see that all pebbles of $\triangle_{m-1,d}$ must be smaller than every pebble of $M_{m-1}$ hence sorting these pebbles independently in the final step gives the desired sorted order.
\end{proof}

%
%
\bibliography{mybibfile}
%

\end{document}